\documentclass{amsart}

\usepackage{amsmath}
\usepackage{amssymb}
\usepackage{enumerate}
\usepackage{slashed}
\usepackage{graphicx}
\usepackage{newlfont}
\usepackage{amsrefs}
\usepackage{comment}

\theoremstyle{plain}
\newtheorem{theorem}{Theorem}[section]
\newtheorem{lemma}[theorem]{Lemma}

\newtheorem{conjecture}[theorem]{Conjecture}
\newtheorem{proposition}[theorem]{Proposition}

\renewcommand{\div}{\operatorname{div}}
\newcommand{\tr}{\operatorname{tr}}
\newcommand{\RN}{Reissner-Nordstr\"om}

\newcommand{\R}{\mathbb{R}}
\newcommand{\E}{\mathbf{E}}
\newcommand{\B}{\mathbf{B}}
\renewcommand{\t}{\mathbf{t}}
\newcommand{\Rb}{\bar R}
\newcommand{\g}{\bar g}
\newcommand{\Db}{\overline\nabla}
\newcommand{\muym}{\mu_{\text{\tiny YM}}}
\newcommand{\madm}{m_{\text{\tiny ADM}}}
\newcommand{\md}{m_{\text{\tiny D}}}

\begin{document}

\title[Extensions of the charged Riemannian Penrose inequality]
{Extensions of the charged \\ Riemannian Penrose inequality}

\author[Khuri]{Marcus Khuri}
\address{Department of Mathematics\\
Stony Brook University\\
Stony Brook, NY 11794, USA}
\email{khuri@math.sunysb.edu}

\author[Weinstein]{Gilbert Weinstein}
\address{Physics Department and Department of Computer Science and Mathematics\\
Ariel University of Samaria\\
Ariel, 40700, Israel}
\email{gilbertw@ariel.ac.il}

\author[Yamada]{Sumio Yamada}
\address{Department of Mathematics\\
Gakushuin University\\
Tokyo, 171-8588, Japan}
\email{yamada@math.gakushuin.ac.jp}

\thanks{M. Khuri acknowledges the support of
NSF Grants DMS-1007156 and DMS-1308753. S. Yamada acknowledges the support of
JSPS Grants 23654061 and 24340009.}

\begin{abstract}
In this paper we investigate the extension of the charged Riemannian
Penrose inequality to the case where charges are present outside the horizon. We
prove a positive result when the charge densities are compactly supported, and present a
counterexample when the charges extend to infinity. We also
discuss additional extensions to other matter models.
\end{abstract}
\maketitle

\section{Introduction} \label{intro}

In~\cite{khuriweinsteinyamadaCRPI}, we proved the Riemmanian Penrose inequality with
charge for multiple black holes.

\begin{theorem} \label{CRPI}
Let $(M,g,E,B)$ be a strongly asymptotically flat initial data
set for the Einstein-Maxwell equations with outermost minimal
surface boundary of area $A=4\pi \rho^2$, with ADM mass $m$, and total charge $q$,
satisfying the charged dominant energy
condition and the Maxwell constraints without charged matter. Then
\begin{equation} \label{upper}
 \rho \leq m + \sqrt{m^2-q^2},
\end{equation}
with equality if and only if the data set arises as
the canonical slice of the Reissner-Nordstr\"{o}m spacetime.
\end{theorem}

Here $M$ is a three dimensional manifold, $g$ a Riemannian metric on $M$, $E$ and $B$
vector fields on
$M$, and $\rho$ is called the \emph{area radius} of the outermost minimal surface.
When $|q|\geq \rho$, the positive
mass theorem with charge~\cite{GibbonsHawkingHorowitzPerry}, $m\geq
|q|$, immediately implies~\eqref{upper}.
Thus to prove Theorem~\ref{CRPI}, it was sufficient to prove~\eqref{upper}
when $|q|\leq \rho$, which was accomplished with the following result.

\begin{theorem} \label{thm-IMCF}
Let $(M,g,E,B)$ be a strongly asymptotically flat initial data set
for the Einstein-Maxwell equations with outermost minimal
surface boundary of area $A=4\pi \rho^2$, with ADM mass $m$, and total charge $q$,
satisfying the charged dominant energy
condition and the Maxwell constraints without charged matter. If $|q|\leq \rho$, then
\begin{equation} \label{mCRPI}
  m\geq \frac12\left( \rho + \frac{q^2}{\rho} \right),
\end{equation}
with equality if and only if the data set arises
as the canonical slice of the Reissner-Nordstr\"{o}m spacetime.
\end{theorem}

We proved this theorem using a conformal flow method adapted
from~\cite{bray2001}. In~\cite{weinsteinyamada}, a counterexample
to~\eqref{mCRPI} based on the Majumdar-Papapetrou solutions was constructed when
$\rho<|q|$, showing that the hypothesis $|q|\leq\rho$ is necessary in
Theorem~\ref{thm-IMCF}.
In this paper we examine whether the hypothesis  $\div E=\div B=0$, i.e.\ the absence of
charges outside the horizon, is necessary. We will prove the following two theorems;
a positive and a negative result.

\begin{theorem} \label{CRPI-charges}
Let $(M,g,E,B)$ be a strongly asymptotically flat initial data set
for the Einstein-Maxwell equations with outermost minimal
surface boundary of area $A=4\pi \rho^2$, with ADM mass $m$, and total charge $q$,
satisfying the charged dominant energy
condition. If $\div E$ and $\div B$ are compactly
supported, then~\eqref{upper} holds.
\end{theorem}

\begin{theorem} \label{counterexample}
There is a spherically symmetric counterexample to~\eqref{upper} with $B=0$, satisfying
all the conditions of Theorem~\ref{CRPI-charges} except that $\div E$ is not
compactly supported.
\end{theorem}

The existence of a spherically symmetric counterexample was conjectured in~\cite{malec}.
Note that in Theorem~\ref{CRPI-charges} the rigidity statement has been omitted. In
fact, we show that a counterexample to rigidity exists.

\begin{proposition} \label{norigidity}
There is a non-trivial arbitrarily small spherically symmetric perturbation of the \RN\
canonical slice $(M,g,E,0)$, which coincides with \RN\ outside of
an annulus, satisfies all the hypothesis of Theorem~\ref{CRPI-charges} in addition to
$|q|\leq\rho$, and
saturates inequality~\eqref{mCRPI}:
\begin{equation} \label{equatlity}
  m = \frac12\left( \rho + \frac{q^2}{\rho} \right).
\end{equation}
\end{proposition}

We also examine the question of whether the conformal flow technique developed
in~\cites{khuriweinsteinyamadaCRPI,khuriweinsteinyamada2013} can be extended to
other matter models. It is expected that for a Riemannian Penrose inequality to hold, the
matter model should admit a unique stable static black hole solution. At the moment, we
are only aware of two such models, the Einstein-Abelian-Yang-Mills (EAYM), and the
Einstein-Maxwell-Dilaton (EMD). We prove a
Riemannian Penrose inequality for EAYM black holes, and conjecture that a
Riemannian Penrose inequality for EMD black holes also holds.

\begin{theorem} \label{abelian-YM}
Let $(M,g,E_1,B_1,\dots,E_\ell,B_\ell)$ be a strongly asymptotically flat
initial data set for the EAYM equations with an
outermost minimal surface boundary of area $A=4\pi\rho^2$, with ADM mass $m$, and total
charge $q$, satisfying the EAYM dominant energy
condition~\eqref{abelian-dec}. If $\div E_i$ and
$\div B_i$ are compactly
supported, then~\eqref{upper} holds. If there are no charges outside the horizon then
equality holds if and only if the data set
arises as the canonical slice of an EAYM Reissner-Nordstr\"{o}m spacetime.
\end{theorem}

\begin{conjecture} \label{dilaton}
Let $(M,g,E,B,\phi)$ be a strongly asymptotically flat initial data
set for the EMD equations with an outermost minimal surface boundary of area
$A=4\pi \rho^2$, ADM mass $m$, and total charge $q$, satisfying the EMD dominant
energy
condition~\eqref{dilaton-dec}. If the charge densities $\div(e^{-2\phi} E)$ and
$\div B$ are compactly supported, then
\begin{equation} \label{dRPI}
  m \geq \frac12\, \sqrt{\rho^2+2q^2}.
\end{equation}
If $e^{-2\phi} E$ and $B$ are divergence free, then equality
holds if and only if the data set arises as the canonical slice of the Gibbons EMD
black hole.
\end{conjecture}

The plan of the paper is as follows. In the next section, we briefly present some of the
background, heuristics, and prior work on the Riemannian Penrose inequality without and
with charge. In that section
we also discuss the dominant energy conditions for EAYM data
sets and for EMD data sets, as well as the EMD black hole. In
Section~\ref{charges}, we prove Theorem~\ref{CRPI-charges},
Theorem~\ref{counterexample} and Proposition~\ref{norigidity}. In
Section~\ref{matter}, we prove Theorem~\ref{abelian-YM} and discuss
Conjecture~\ref{dilaton}.

\section{Background} \label{background}

\subsection{The Penrose inequality}
The Penrose inequality was originally proposed by Penrose as a test of the cosmic
censorship conjecture, or more generally a test of the standard picture of gravitational
collapse~\cites{penrose1973,penrose1982}. Consider a strongly asymptotically flat
Cauchy surface in a spacetime satisfying
the dominant energy condition, with ADM mass $m_0$, and containing an
event horizon of area $A_0=4\pi \rho_0^2$, which undergoes
gravitational collapse and settles to a Kerr-Newman black hole. Since the ADM mass
$m$ of the final state is no greater
than $m_0$, and since the end state area radius $\rho$ is no less than $\rho_0$ by
Hawking's
area theorem~\cite{hawkingellis}, and since for the final state
$m\geq \rho/2$ in order to avoid naked singularities,
it must have been the case that $m_0\geq \rho_0/2$ also at the
beginning of the evolution. The event horizon is indiscernible in the original slice
without knowing the full evolution.
However, one may replace the event horizon by the outermost minimal area enclosure of the
apparent horizon, the boundary
of the region admitting trapped surfaces, and reach the same conclusion. A
counterexample to the Penrose inequality would therefore have suggested Cauchy data which
leads under the Einstein evolution to naked singularities, while a proof of
the inequality could be viewed as evidence in support of cosmic censorship.

\subsection{Penrose's heuristic argument for multiple black holes} \label{penrose}
It is usually assumed that the end state of gravitational
collapse is a single Kerr-Newman black hole.
However, a more appropriate assumption for the end state is a
finite number of mutually distant Kerr-Newman black holes moving apart with
asymptotically
constant velocity. This should be the result, if for instance, two distant black
holes were initially moving away from each other sufficiently fast. We will
now describe the heuristic Penrose argument for the charged Penrose
inequality~\eqref{upper} in this setting. It appears that this has not been
previously considered in the literature.

Let $m_{i}$, $A_{i}=4\pi\rho_i^2$,
$q_{i}$, $J_i$ denote the ADM masses, horizon areas, total charges, and angular momenta of the end state black
holes. Then the total (ADM) mass, horizon area radius, and charge of the end state is
$m=\sum m_i$, $\rho=(\sum\rho_i^2)^{1/2}$, $q=\sum q_i$.
The area radius of the Kerr-Newman black hole~\cite{dainkhuriweinsteinyamada} is given by
\[
 \frac{\rho_i^2}{2} = m_i^2-\frac{q_i^2}{2} + \sqrt{\left(m_i^2-\frac{q_i^2}{2}\right)^2-\frac{q_i^4}{4}-J_i^2}
 \leq m_i^2-\frac{q_i^2}{2} + \sqrt{\left(m_i^2-\frac{q_i^2}{2}\right)^2-\frac{q_i^4}{4}}.
\]
It follows that
\[
  \rho_i \leq m_i+\sqrt{m_i^2-q_i^2}.
\]
Let $m_{0}$, $\rho_{0}$, $q_{0}$ denote the ADM mass, horizon area radius, and total
charge of an initial state. Under the
assumption that no charged matter is present, the total charge is conserved
$q_0=q=\sum q_i$. Moreover, by the Hawking area theorem
$\rho_0\leq\rho\leq\sum \rho_i$, and since gravitational waves may only carry away
positive energy $m_0 \geq m = \sum m_i$.

\begin{lemma}
Let $a_i$, $b_i$ be real numbers, then
\begin{equation} \label{hardy}
  \left(\sum a_i \right)^2 + \left(\sum b_i \right)^2 \leq \left( \sum\sqrt{a_i^2 +
  b_i^2} \right)^2
\end{equation}
\end{lemma}
\begin{proof}
This follows directly from the triangle inequality for $n$ points in $\mathbb{R}^{2}$.
\end{proof}
Now let $a_i=\sqrt{m_i^2-q_i^2}$ and $b_i=q_i$, then we have
\[
  \left(\sum\sqrt{m_i^2-q_i^2}\right)^2 + \left(\sum q_i\right)^2 \leq
  \left(\sum m_i\right)^2,
\]
or equivalently
\[
  \sum\sqrt{m_i^2-q_i^2} \leq \sqrt{\left(\sum m_i\right)^2 - \left(\sum q_i\right)^2}  =
  \sqrt{m^2-q^2}.
\]
Thus we conclude:
\[
  \rho_0\leq\rho \leq \sum \rho_i \leq \sum m_i + \sum\sqrt{m_i^2-q_i^2} \leq
  m + \sqrt{m^2-q^2} \leq m_0 + \sqrt{m_0^2-q_0^2}.
\]

\subsection{The Riemannian Penrose inequality} \label{section-RPI}
The Penrose inequality further simplifies in the time-symmetric case, where the apparent
horizon
coincides with the outermost minimal surface. Moreover, the dominant energy condition
reduces
now to non-negative scalar curvature of the Cauchy hypersurface, leading to the
Riemannian version of the inequality.

\begin{theorem} \label{RPI}
Let $(M,g)$ be a three dimensional strongly asymptotically flat Riemannian manifold of
non-negative
scalar curvature, with an outermost minimal surface boundary $\Sigma$ of area $A=4\pi
\rho^2$, then
\[
  m\geq \frac12 \rho,
\]
with equality if and only if the manifold is a canonical slice of the Schwarzschild
spacetime.
\end{theorem}

The first published proof of this theorem was given by Huisken-Ilmanen
in~\cite{huiskenilmanen2001}
using inverse mean curvature flow, but required the assumption that $\Sigma$ be
connected. Another proof by Bray in~\cite{bray2001} used a conformal flow and
applied more generally to a non-connected boundary $\Sigma$. Nonetheless, one advantage
of the first proof was that it could be immediately generalized to obtain the
Einstein-Maxwell case~\cite{jang1979},
i.e.\ to prove Theorem~\ref{thm-IMCF} in the case of a connected $\Sigma$.
It is important to point out that~\eqref{mCRPI} implies both~\eqref{upper} and
a lower bound:
\begin{equation} \label{lower}
  \rho \geq m - \sqrt{m^2-q^2}.
\end{equation}
In fact this lower bound follows, when $\Sigma$ is connected, from the stability of
the outermost horizon~\cites{khuriweinsteinyamadaPRD,dainkhuriweinsteinyamada}. Indeed,
the absence of charges outside the horizon, the Cauchy-Schwartz inequality,
the charged dominant energy condition, the stability of the outermost horizon, and
the Gauss-Bonnet Theorem, imply
\begin{equation} \label{stability}
\begin{gathered}
  \frac{ 4\pi q^2}{\rho^2}
  = \frac1{4\pi\rho^2} \left[ \left(\int_{\Sigma} g(E,n) dA\right)^2
  + \left(\int_{\Sigma} g(B,n) dA\right)^2 \right] \\
  \leq \int_{\Sigma} \left( |E|^2 + |B|^2 \right) dA
  \leq \int_\Sigma \frac{1}{2}R\,dA \leq \int_\Sigma K = 4\pi N,
\end{gathered}
\end{equation}
where $R$ is the scalar curvature, and $N$, $K$, and $n$ are respectively the number of
components, the Gauss curvature, and the unit normal of $\Sigma$. In
particular one obtains the area-charge inequality $|q|\leq \rho$, or more generally
\begin{equation} \label{n-bh}
  |q| \leq \sqrt N \rho.
\end{equation}
This inequality was obtained in~\cite{dainjaramilloreiris2012}. If $N=1$, the lower
bound~\eqref{lower} follows immediately:
\[
   m = \sqrt{q^2+m^2-q^2} \leq |q| + \sqrt{m^2-q^2} \leq \rho + \sqrt{m^2-q^2}.
\]

\begin{figure}
\includegraphics[width=11cm]{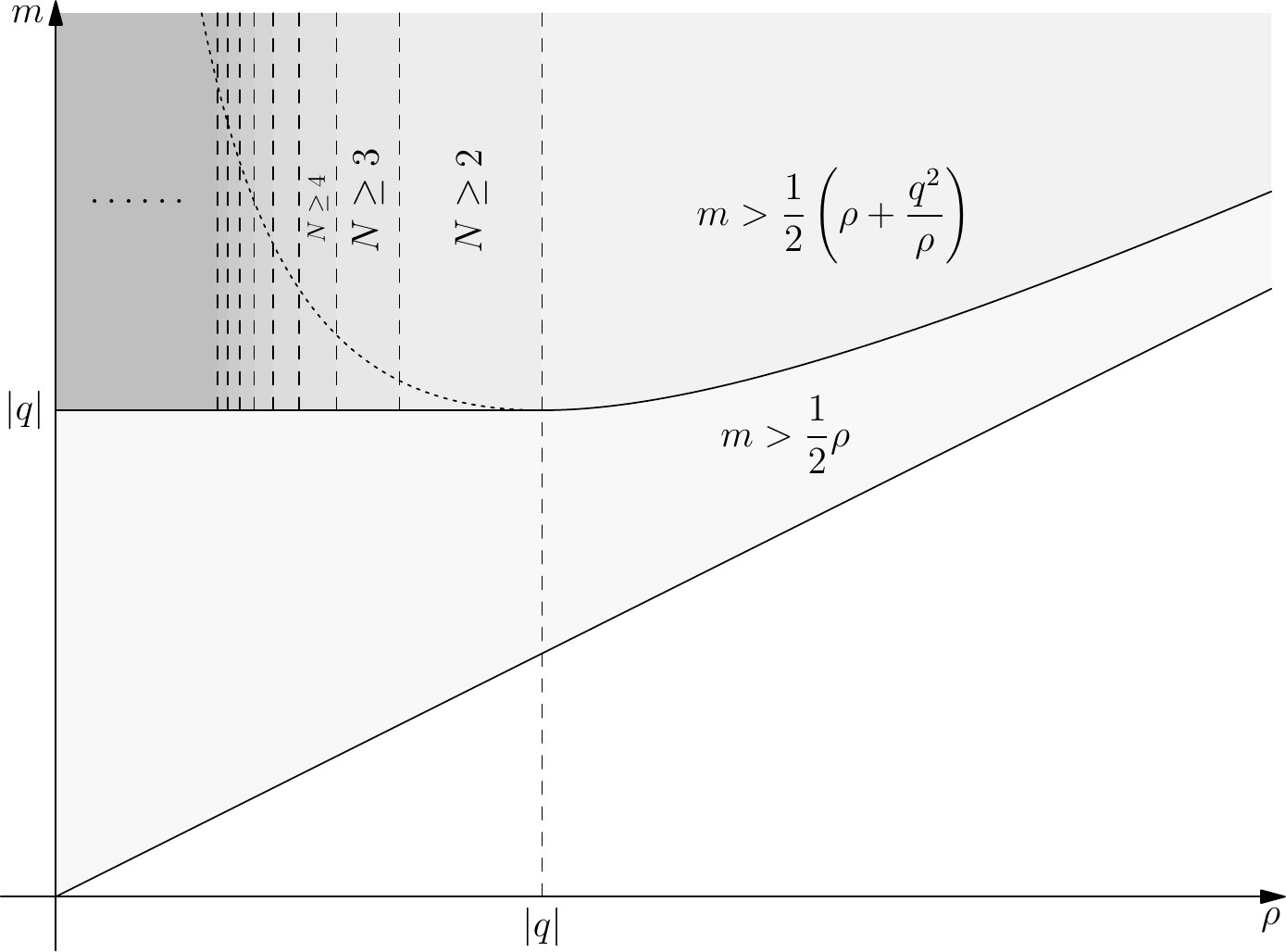}
 \caption{Graphical representation of geometric inequalities}  \label{graphical}
\end{figure}

We note that the inequality~\eqref{n-bh} implies that if $|q|>\sqrt N \rho$, then there
must be at least $N$ black holes, demonstrating that inequalities between these geometric
quantities have topological consequences. All these results are summarized in
Figure~\ref{graphical}.

\subsection{Data for the Einstein-Abelian-Yang-Mills equations}
The gauge group for the EAYM matter model is
$\bigotimes_{\ell} U(1)$. The connection, and the field strength $F$, have values in the
abelian Lie algebra $\R^{\ell}$. As usual, the electric and magnetic fields are given by
$\E=(E_1,\dots,E_\ell)=i_{\t}F$, $\B=(B_1,\dots,B_\ell)=i_{\t}{\ast F}$, where $\t$ is
the future directed timelike unit normal to the initial data slice, and $\ast$ is the
Hodge dual. The EAYM equations of motions are
\begin{gather*}
  G_{\mu\nu} = \Rb_{\mu\nu} - \frac12 \Rb\g_{\mu\nu} = 8\pi T_{\mu\nu}, \\
  T_{\mu\nu} = \frac{1}{4\pi} \left( F_{\mu\alpha}F_\nu{}^{\alpha} - \frac14
  F_{\alpha\beta} F^{\alpha\beta}   \g_{\mu\nu}\right)
  = \frac1{8\pi}   \left(F_{\mu\alpha}F_\nu{}^{\alpha} +
  {\ast F}_{\mu\alpha} {\ast F}_{\nu}{}^\alpha\right), \\
  \Db^\mu F_{\mu\nu}=0, \qquad  \Db^\mu {\ast F}_{\mu\nu} = 0,
\end{gather*}
where barred quantities refer to spacetime objects, and Greek indices run $0,\dots,3$.
From the expression for the energy-momentum-stress tensor, we obtain the energy
density of the slice after contributions from the Yang-Mills
fields have been removed:
\[
  8\pi\muym = \Rb_{\t\t} + \frac12 \Rb - 8\pi T_{\t\t} =
  \frac12 \left( R - |k|^2 + (\tr k)^2 \right) - \left( \sum_{i=1}^\ell |E_i|^2 +
  |B_i|^2\right),
\]
where $k$ is the second fundamental form of the slice. In the time-symmetric case, $k=0$,
and the dominant energy condition is
\begin{equation} \label{abelian-dec}
 R\geq 2 \left( \sum_{i=1}^\ell |E_i|^2 + |B_i|^2\right).
\end{equation}
We note that the dominant energy condition is usually stronger, namely $\muym \geq
|J|$ where $J_i=\frac1{8\pi} G_{\t i} -
T_{\t i}$, while here we only use the weaker condition $\muym\geq0$. The total charges
are defined by
\begin{gather*}
  q_{E_i} = \frac1{4\pi} \int_{S_\infty} g(E_i,n)\, dA, \quad
  q_{B_i} = \frac1{4\pi} \int_{S_\infty} g(B_i,n)\, dA, \\
  q = \sqrt{q_{E_1}^2 + q_{B_1}^2 + \dots + q_{E_\ell}^2 + q_{B_\ell}^2}.
\end{gather*}

\subsection{Data for the Einstein-Maxwell-Dilaton equations} \label{dPI}
The EMD action has the Lagrangian
\[
  \cal L = \left(\Rb - 2\Db_\alpha\phi \Db^\alpha \phi - e^{-2\phi}
  F_{\alpha\beta} F^{\alpha\beta} \right) \sqrt{-\det \g}.
\]
There are other EMD models where the coupling $e^{-2\phi}$ is replaced by $e^{-2a\phi}$, in which $a$ is a coupling constant. While it is possible to conjecture a Riemannian Penrose inequality also for $a\ne1$, we leave
this for future work. From the Lagrangian, one obtains the EMD equations of motion
\begin{gather*}
  G_{\mu\nu} = 8\pi T_{\mu\nu}, \\
  T_{\mu\nu} = \frac{e^{-2\phi}}{8\pi}   \left(F_{\mu\alpha}F_\nu{}^{\alpha} +
  {\ast F}_{\mu\alpha} {\ast F}_{\nu}{}^\alpha\right)
  - \frac1{4\pi} \left( \Db_\mu\phi \Db_\nu\phi - \frac12
  \Db_\alpha\phi \Db^\alpha\phi\, \g_{\mu\nu} \right), \\
  \Db^\mu \left( e^{-2\phi} F_{\mu\nu} \right)=0, \qquad
  \Db^\mu {\ast F}_{\mu\nu} = 0, \qquad
  \Box_{\g} \phi + \frac12 e^{-2\phi} F_{\mu\nu} F^{\mu\nu} = 0.
\end{gather*}
As above
\[
 8\pi\mu_{\text{\tiny MD}} = \frac12 \left( R - |k|^2 + (\tr k)^2 \right) - e^{-2\phi}
\left( |E|^2
 + |B|^2\right) - (\partial_t\phi)^2 - |\nabla\phi|^2.
\]
In the time-symmetric case, $k=0$ and $\partial_t\phi=0$, thus from $\mu_{\text{\tiny
MD}}\geq0$ we
get the EMD dominant energy condition
\begin{equation} \label{dilaton-dec}
 R \geq 2 e^{-2\phi} \left( |E|^2  + |B|^2\right) + 2|\nabla\phi|^2.
\end{equation}
The charges contained in a surface $S$ are given by
\begin{gather*}
  q_E(S) = \frac1{4\pi} \int_{S} e^{-2\phi} g(E,n)\, dA, \quad
  q_B(S) = \frac1{4\pi} \int_{S} g(B,n)\, dA, \\
  q(S) = \sqrt{q_E(S)^2 + q_B(S)^2}.
\end{gather*}
We note that if $\div(e^{-2\phi} E)=\div B=0$ these charges depend only on the
homology class of $S$. The total charges are obtained by taking $S=S_\infty$.

A static spherically symmetric EMD black hole was discovered by
Gibbons~\cites{gibbons1982,gibbonsmaeda}:
\begin{equation} \label{DBH}
\begin{gathered}
  ds^2 = -\left(1-\frac{2m}{r}\right) dt^2 + \left(1-\frac{2m}{r}\right)^{-1} dr^2 +
  r\left(r-\frac{q^2}{m}\right) d\omega^2, \\
  F_{rt} = \frac{q}{r^2}, \qquad e^{2\phi} = 1 - \frac{q^2}{mr},
\end{gathered}
\end{equation}
where $d\omega^2$ is the round metric on the unit 2-sphere. It was proven to be the
unique static solution in~\cites{masood1993,marssimon2002}.
We note that the scalar curvature of this EMD black hole is given by
\[
  R = \frac1{r(r-q^2/m)} \left( \frac{2q^2}{r^2}
      + \left(1-\frac{2m}{r}\right) \frac{q^4}{2m^2r(r-q^2/m)} \right),
\]
and hence~\eqref{dilaton-dec} is an equality in this case. For more details on this
solution, see~\cite{horowitz}.

\section{Charges outside the horizon} \label{charges}

\begin{proof}[Proof of Theorem~\ref{CRPI-charges}]
The proof follows directly
from~\cites{khuriweinsteinyamada2013,khuriweinsteinyamadaCRPI}, after
noting that the only place which requires $E$ and $B$ to be
divergence free is when one applies the inverse mean curvature flow at the end of the
proof. According to
the exhaustion result the flowing surfaces $\Sigma_{t}$ eventually become connected
and enclose any large coordinate sphere. Thus, if $E$ and $B$ are divergence free outside
of a large coordinate sphere $S_{r}$, then we may apply the inverse mean curvature flow
argument once $\Sigma_{t}$ encloses $S_{r}$.
\end{proof}

\begin{proof}[Proof of Theorem~\ref{counterexample}]
Let the spherically symmetric metric be given by
\[
  ds^2 = \left(1-\frac{2m(r)}{r}\right)^{-1} dr^2 + r^2 \, d\omega^2.
\]
It is easy to check that $m(r)$ is
the Hawking mass of the coordinate sphere $S_r$. Furthermore, a
horizon occurs at $r=r_0$ if
$m(r_0)=r_0/2$. If $m(r)$ is increasing there is only one horizon, hence $S_{r_0}$
is the outermost horizon. We will use the charged Hawking
mass~\cite{disconzikhuri}
\begin{equation} \label{chm}
  m_{c}(r) = m(r)+\frac{q^2}{2r},
\end{equation}
where $q$ is the total charge. Note that on the horizon, we have
\[
  m_{c}(r_0) = \frac{1}{2} \left(r_0 + \frac{q^2}{r_0} \right),
\]
and at infinity $m_{c}$ tends to the ADM mass $\madm$. Clearly, if $m_c'<0$ for
\addtocounter{footnote}{1}
$r>r_0$, then~\eqref{mCRPI} will be violated. A violation of~\eqref{mCRPI} with
$|q|\leq\rho=r_0$ implies a violation of~\eqref{upper}. Thus, in order to
construct a counterexample it suffices to find two functions $m(r)$ and
$f(r)$, where $E=f \partial_r$, satisfying the following conditions:
\begin{enumerate}[A.]
 \item The charged dominant energy condition
\begin{equation} \label{cdec}
  R = \frac{4m'}{r^2} \geq 2|E|^2 = \frac{2f^2}{1-\frac{2m}{r}}.
\end{equation}
 \item Total charge is $q$: $r^2 f \to q$ as $r\to\infty$.
 \item The condition $|q|\leq r_0$.
\item Asymptotic flatness: $m\to \madm$ as $r\to\infty$.
 \item  The inequality
\begin{equation} \label{monotone}
  m_c'=m'-\frac{q^2}{2r^2} <0.
\end{equation}
\end{enumerate}

Choose $0<q<r_0$. We begin by solving
\[
  m' = \frac{q^2}{2r^2}\left(1-\frac{r_0}{r}\right), \quad m(r_0) = \frac{r_0}{2},
\]
so that $m'>0$ and~\eqref{monotone} is satisfied for $r>r_0$. One finds
\[
  m=\frac{r_0}{2} + \frac{q^2}{4r_0} - \frac{q^2}{2r} + \frac{q^2r_0}{4r^2}.
\]
Clearly asymptotic flatness is satisfied, and in fact
\[
  \madm = \frac{r_0}{2} + \frac{q^2}{4r_0}.
\]
Furthermore $2m'-1<0$, and consequently $2m<r$, for $r>r_0$.
Finally observe that
\[
  h(r) := \frac{2m'}{r^2}\left(1-\frac{2m}{r}\right) =
  \frac{q^2}{r^4}\left(1-\frac{r_0}{r}\right)
  \left(1-\frac{2m}{r}\right).
\]
It follows that we can find a function $f(r)$ satisfying $f(r_0)=0$, the
charged dominant energy condition $f(r)^2<h(r)$ for $r>r_0$, and $r^2f(r)\to q$ as
$r\to\infty$.
\end{proof}

\begin{proof}[Proof of Proposition~\ref{norigidity}]
Let $(M,g_{0},E_{0})$ denote the canonical slice of the Reissner-Nordstr\"{o}m spacetime,
and let $r$ denote the anisotropic radial coordinate for this data. Consider the annulus
$\Omega(r_{1},r_{2})$ where $r_{2}>r_{1}>m+\sqrt{m^{2}-q^{2}}$, and fix $f\in
C_{c}^{\infty}(\Omega(r_{1},r_{2}))$ to be nonnegative
and not identically zero. Since the scalar curvature $R_{g_{0}}$ is nonconstant in
$\Omega(r_{1},r_{2})$, the formal $L^{2}$-adjoint of the linearized scalar curvature
operator has trivial kernel. It then follows from \cite{corvino2000} that there exists a
smooth contravariant 2-tensor $h\in C_{c}^{\infty}(\Omega(r_{1},r_{2}))$, such that the
scalar curvature of the metric $\tilde{g}=g_{0}+\varepsilon h$ is given by
$R_{\tilde{g}}=R_{g_{0}}+\varepsilon f$ for $\varepsilon>0$ sufficiently small. It is
also clear from the proof in \cite{corvino2000} that if $f$ is chosen to be spherically
symmetric, then $h$ is also spherically symmetric.

Define $\tilde{E}=E_{0}+\varepsilon V$, and observe that
\begin{equation}\label{2}
|\tilde{E}|_{\tilde{g}}^{2}=|E_{0}|_{g_{0}}^{2}+2\varepsilon
g_{0}(E_{0},V)+\varepsilon h(E_{0},E_{0})
+O(\varepsilon^{2}).
\end{equation}
Since $E_{0}=-q\nabla r^{-1}$, the implicit function theorem may be used to find $V$ such
that $|\tilde{E}|_{\tilde{g}}^{2}=|E_{0}|_{g_{0}}^{2}$. Therefore
\begin{equation}\label{3}
R_{\tilde{g}}=R_{g_{0}}+\varepsilon f= 2|E_{0}|_{g_{0}}^{2}+\varepsilon f
=2|\tilde{E}|_{\tilde{g}}^{2}+\varepsilon f\geq2|\tilde{E}|_{\tilde{g}}^{2},
\end{equation}
so that the initial data set $(M,\tilde{g},\tilde{E})$ satisfies the charged
dominant energy condition.
Moreover, it is clear that this initial data set satisfies all desired hypotheses, except
possibly the outermost
condition for the minimal surface boundary. However, it can easily be seen that the
minimal boundary is outermost by choosing $\varepsilon$ sufficiently small. Namely, by
choosing $\varepsilon$ sufficiently small, the mean curvature of each coordinate sphere
$S_{r}$, $r>m+\sqrt{m^{2}-q^{2}}$, remains positive. Thus, the exterior region is
foliated
by surfaces of positive mean curvature, showing that no other minimal surface enclosing
the boundary can exist (by the maximum principle for minimal surfaces).
\end{proof}

\section{Other matter models} \label{matter}

\begin{proof}[Proof of Theorem~\ref{abelian-YM}]
Apply a rotation
\[
 (E_1,B_1,\dots,E_\ell,B_\ell) \mapsto (\tilde E_1,\tilde
 B_1,\dots,\tilde E_\ell,\tilde B_\ell),
\]
so that all the transformed charges $q_{\tilde E_i}$
and $q_{\tilde B_i}$ vanish except possibly $q_{\tilde E_1}$, and $q_{\tilde E_1} = q$.
If $q\geq \rho$ then the positive mass theorem with charge applies, and we
get~\eqref{upper} as explained in the introduction. Otherwise, starting with this new
data, one may apply the proof
in~\cite{khuriweinsteinyamadaCRPI}, conformally deforming all the electromagnetic fields
in the same manner $\tilde E_i(t)=u_t^{-6} \tilde E_i$, $\tilde B_i(t) = u_t^{-6} \tilde
B_i$. Both the charged conformal flow  and the inverse mean curvature flow
arguments now proceed as in Theorem~\ref{CRPI-charges},
and we obtain~\eqref{mCRPI}. If there are no charges and equality holds, we obtain as
in~\cite{khuriweinsteinyamadaCRPI} that the transformed data is a canonical
slice of \RN. After rotating back, we find that the original data is a canonical
slice of an EAYM \RN\ black hole.
\end{proof}

\begin{proof}[Discussion on Conjecture~\ref{dilaton}]
An initial data set for the EMD equations consists of $(M,g,k,E,B,\phi,\psi)$, where
$\psi$ is the initial time derivative of $\phi$. In the time-symmetric case $k$ and
$\psi$ vanish. We assume the data is strongly asymptotically flat, satisfies the EMD
dominant energy
condition~\eqref{dilaton-dec}, and that charges are absent,
$\div(e^{-2\phi}E)=\div B=0$ outside the horizon, or more
generally outside a compact set.

Our first observation is that a Penrose heuristic argument, as in Section~\ref{penrose},
is still valid. Start from data with paramaters $m_0$, $\rho_0$, $q_0$, and
assume the data settles to an EMD black hole as in~\eqref{DBH} with parameters $m$,
$\rho$, and $q$. Then, the charge is conserved $q=q_0$, the Hawking area theorem applies
giving $\rho_0\leq\rho$, and as before the masses satisfy $m\leq m_0$. Since for the
EMD black hole we have
\[
  m=\frac12 \sqrt{\rho^2+2q^2},
\]
it follows that
\[
  m_0 \geq m = \frac12 \sqrt{\rho^2+2q^2} \geq \frac12 \sqrt{\rho_0^2 + 2q_0^2},
\]
and hence~\eqref{dRPI} holds for the initial data.

We now point out a number of significant differences between this setting
and our work in~\cite{khuriweinsteinyamadaCRPI}. First, there are analogs of
the Majumdar-Papapetrou multiple neck solutions~\cite{garfinklehorowitzstrominger}, but
they do not pose an obstruction to the lower bound~\eqref{dRPI}. Indeed, the right
hand side of~\eqref{dRPI} is always monotone in $\rho$ unlike the right hand side
of~\eqref{mCRPI}, hence we do not expect an auxiliary inequality, such as the
{area/charge}
inequality in Theorem~\ref{thm-IMCF}, to be necessary to prove~\eqref{dRPI}. These
extremal solutions have, in a conformally related ``string'' metric,  asymptotically
cylindrical infinitely long necks just as in Majumdar-Papapetrou, but in the spacetime
metric these have cross sectional area tending to zero at a finite distance.
Since we have $\rho=0$ in this case, with $q\ne0$, there are EMD black holes with
nearby parameters that have a single connected component horizon, and charge to area
radius ratio arbitrarily small. Thus the
phenomenon described in Section~\ref{section-RPI} whereby excess charge leads to high
multiplicity of horizon components is absent in the EMD model. The
mathematical reason is that in the EMD case, to replace
the integral over $S$ with the integral over $S_\infty$ in
inequality~\eqref{stability}, would
require $e^{-\phi}E$ and $e^{-\phi} B$  to be divergence free rather than, as
is assumed, $e^{-2\phi} E$ and $B$.

Notwithstanding these observations, the many similarities we now outline, lead us to
surmise that the same approach used in~\cite{khuriweinsteinyamadaCRPI} should apply in
the EMD setting. First, a positive mass theorem with charge holds~\cites{rogatko,nozawa}.

\begin{theorem} \label{dilaton-PMT}
Let $(M,g,k,E,B,\phi,\psi)$ be a strongly asymptotically flat initial data set
for the EMD equations with an apparent horizon boundary, satisfying the
EMD dominant energy condition, with ADM mass $m$ and total charge $q$. If charge
densities $\div(e^{-2\phi} E)$ and
$\div B$ vanish, then
\[
  m \geq \frac1{\sqrt{2}} |q|.
\]
\end{theorem}

Observe that the initial data is not assumed to be time-symmetric. On the other hand, the
rigidity statement has not yet been proved. We suspect that in analogy with the EM
case, equality holds if and only if the data can be embedded as a Cauchy slice in one of
the GHS static extremal EMD solutions from~\cite{garfinklehorowitzstrominger}, and if the
data is time-symmetric equality holds if and only if that slice is the cannonical slice.
That said, contrary to~\cite{khuriweinsteinyamadaCRPI}, Theorem~\ref{dilaton-PMT}
likely plays no role in the proof of~\eqref{dRPI}.
Instead, we believe that an EMD charged Hawking mass $\md$ exists as
in~\cite{disconzikhuri} satisfying the following properties: (i) $\md$ is monotonically
non-decreasing under the Huisken-Ilmanen weak inverse mean curvature flow, and it is
constant only for the central spheres in the spherically symmetric EMD black hole; (ii)
$\md$ converges to $\frac12 \sqrt{\rho^2+2q^2}$ on the outermost horizon, and to the ADM
mass at infinity. This would prove Conjecture~\ref{dilaton}
including the rigidity statement, provided the horizon is connected and there are
no charges outside the horizon. Finally, we conjecture that there exists an EMD conformal
flow satisfying: (i) the EMD dominant energy condition is preserved under the flow; (ii)
the absence of charges is preserved under the flow; (iii) the ADM mass is monotonically
non-increasing under the flow and constant only for the spherically symmetric EMD
black holes; (iv) the area of
the boundary is constant under the flow; (v) the boundary eventually encloses any compact
set on the original manifold. This would prove the full conjecture. For (iii), a critical
ingredient, used in both~\cite{bray2001} and~\cite{khuriweinsteinyamadaCRPI}, namely a
conformal factor used to prove uniqueness in a doubling argument, is already
available~\cites{masood1993,marssimon2002}.
\end{proof}

\bibliographystyle{}


\end{document}